\documentclass[a4paper,english]{amsart}
\usepackage[T1]{fontenc}
\usepackage[latin9]{inputenc}
\usepackage{verbatim}
\usepackage{textcomp}
\usepackage{amsthm}
\usepackage{amssymb}
\usepackage{graphicx}

\makeatletter

\pdfpageheight\paperheight
\pdfpagewidth\paperwidth

\providecommand{\tabularnewline}{\\}

\numberwithin{equation}{section}
\numberwithin{figure}{section}
\theoremstyle{plain}
\newtheorem{thm}{\protect\theoremname}
  \theoremstyle{plain}
  \newtheorem{lem}[thm]{\protect\lemmaname}
  \theoremstyle{plain}
  \newtheorem{conjecture}[thm]{\protect\conjecturename}
  \theoremstyle{plain}
  \newtheorem{prop}[thm]{\protect\propositionname}

\usepackage{lscape}

\usepackage{babel}
\providecommand{\conjecturename}{Conjecture}
  \providecommand{\lemmaname}{Lemma}
  \providecommand{\propositionname}{Proposition}
\providecommand{\theoremname}{Theorem}

\makeatother

\usepackage{babel}
  \providecommand{\conjecturename}{Conjecture}
  \providecommand{\lemmaname}{Lemma}
  \providecommand{\propositionname}{Proposition}
\providecommand{\theoremname}{Theorem}

\begin{document}

\title{$BCCB$ complex Hadamard matrices of order 9, and MUBs}

\author{Bengt R Karlsson}

\address{Uppsala University, Department of Physics and Astronomy, Box 516,
SE-751 20, Uppsala, Sweden}

\email{bengt.karlsson@physics.uu.se}
\begin{abstract}
A new type of complex Hadamard matrices of order 9 are constructed.
The studied matrices are symmetric, block circulant with circulant
blocks ($BCCB$) and form an until now unknown non-reducible and non-affine
two-parameter orbit. Several suborbits are identified, including a
one-parameter intersection with the Fourier orbit $F_{9}^{(4)}$.
The defect of this new type of Hadamard matrices is observed to vary,
from a generic value 2 to the anomalous values 4 and 10 for some sub-orbits,
and to $12$ and 16 for some single matrices. The latter matrices
are shown to be related to complete sets of MUBs in dimension 9. 
\end{abstract}



\maketitle

\section{Introduction}

Complex Hadamard matrices have turned out hard to describe in a uniform
manner and a comprehensive understanding of such matrices has only
been achieved in orders $N\le5$. At higher orders, the number of
known complex Hadamard matrices, or orbits of such matrices, is growing
(see the catalogue in \cite{Krakow webguide}), but no general construction
or classification principle has emerged. For example, almost all known
complex Hadamard matrices are either isolated, or elements in affine
orbits stemming from a seed matrix, typically but not exclusively
the Fourier matrix $F_{N}$. However, non-affine orbits have also
been discovered, and in order 6 such orbits play a major role. As
another example, prime order Hadamard matrices might be thought of
as more elementary than those of composite order. Indeed, almost all
known complex Hadamard matrices of orders $4,6,8,9,10$ and $12$
can be seen as composed from $F_{2},F_{3}$ and $F_{5}$. However,
the most general, non-affine orbit in order 6 is not composed of $F_{2}$
and/or $F_{3}$. For an overview of complex Hadamard matrices, see
\cite{Krakow webguide,Tadej guide,Bengtsson et al,Szollosi Thesis}.

In view of this situation it is of some interest to identify and categorize
as many different (orbits of) complex Hadamard matrices as possible.
As a contri\-bu\-tion to these efforts, in this paper we report
a new orbit in order 9, of a kind not encountered before and therefore
of relevance for the ongoing efforts to better understand the full
set of complex Hadamard matrices.

\section{\label{sec:Notation-and-definitions.}Notation and definitions.}

A complex Hadamard matrix $H_{N}$ is an $N\times N$ matrix with
complex elements of modulus $1$, and such that $H_{N}^{\dagger}H_{N}=NI$
(the unitarity constraint). Here, $I$ is the identity matrix. In
this paper, all matrices referred to as Hadamard will be of this kind.
Hadamard matrices exist for any $N$, as exemplified by the Fourier
matrix with elements $(F_{N})_{ij}=\omega_{N}^{(i-1)(j-1)}$, with
$\omega_{N}=\exp(2\pi{\rm i}/N)$.

Two Hadamard matrices $H$ and $\tilde{H}$ are said to be equivalent,
$H\sim\tilde{H}$, if there exist diagonal, unitary matrices $D_{1},$
$D_{2}$, and permutation matrices $P_{1},P_{2}$ such that $H=P_{1}D_{1}\tilde{H}D_{2}P_{2}$.
For each Hadamard matrix there is an equivalent dephased matrix with
all elements in the first row and first column equal to $1$.

In higher orders, several-parameter orbits of Hadamard matrices are
prevalent. Two such orbits will be considered equivalent if for each
matrix in one there is an equivalent matrix in the other. A matrix
which is not part of an orbit is termed isolated.

An orbit of dephased complex Hadamard matrices is affine if the phases
of the elements are linear functions of the orbit parameters.

A Hadamard matrix of even order is $H_{2}$-reducible \cite{Karlsson H_2-red}
if it is equivalent to a matrix where all the $2\times2$ submatrices
are also (in general enphased) Hadamard matrices. More generally,
a composite order Hadamard matrix is reducible if it can be seen as
built from Hadamard submatrices of order $2$ or more.

In a circulant matrix \cite{Davis}, each row is a copy of the previous
row shifted one step to the right, with wrap around. The columns of
$F_{N}^{\dagger}$ are eigenvectors of any such $N\times N$ circulant
matrix. If in a circulant matrix $C_{n_{1}}=\mathrm{circ}(a^{(1)},a^{(2)},...,a^{(n_{1})})$
of order $n_{1}$ the elements $a^{(i)}$ are replaced by order $n_{2}$
circulant submatrices $A_{n_{2}}^{(i)}$, the result is an order $N=n_{1}n_{2}$
block circulant with circulant blocks ($BCCB$) matrix, which has
the columns of $F_{n_{1}}^{\dagger}\otimes F_{n_{2}}^{\dagger}$ as
eigenvectors.

The defect \cite{Krakow webguide,Tadej guide} of a unitary $N\times N$
matrix with elements $H_{ij}=\exp({\rm i}tR_{ij})$ equals $d(H)=r-(2N-1)$,
where $r$ is the dimension of the solution space for $\frac{d}{dt}(H^{\dagger}H)=0$.
In a $p$-parameter orbit of Hadamard matrices, $d(H)\ge p$, and
its generic defect is the smallest defect encountered along the orbit.
There may exist suborbits with larger generic defect, and in particular
individual matrices with a larger defect.

\section{Low order Hadamard matrices}

In order 2, all complex Hadamard matrices are equivalent to the Fourier
matrix on standard or circulant form, 
\begin{equation}
F_{2}=\left(\begin{array}{cc}
1 & 1\\
1 & -1
\end{array}\right)\,\,\,\,\,\,\,\mathrm{or}\,\,\,\,\,\,C_{2}=\left(\begin{array}{cc}
1 & {\rm i}\\
{\rm i} & 1
\end{array}\right)\sim F_{2}.
\end{equation}
This matrix is isolated.

Also in order 3, all Hadamard matrices are equivalent to the Fourier
matrix on standard or circulant form, 
\begin{equation}
F_{3}=\left(\begin{array}{ccc}
1 & 1 & 1\\
1 & \omega & \omega^{2}\\
1 & \omega^{2} & \omega
\end{array}\right)\,\,\,\,\,\,\,\mathrm{or}\,\,\,\,\,\,C_{3}=\left(\begin{array}{ccc}
1 & \omega & \omega\\
\omega & 1 & \omega\\
\omega & \omega & 1
\end{array}\right)\sim F_{3}\label{eq:F3 and C3}
\end{equation}
where $\omega\,(=\omega_{3})=\exp(2\pi{\rm i}/3)$. Again, this matrix
is isolated.

In order 4, all Hadamard matrices are equivalent to an element in
a one-parameter affine orbit, the Fourier orbit $F_{4}^{(1)}(a)$
in the notation of \cite{Krakow webguide}. If written on manifestly
reducible form, or circulant and $BCCB$ form, this orbit can be represented
by 
\begin{equation}
\left(\begin{array}{cc}
F_{2} & \Delta F_{2}\\
F_{2} & -\Delta F_{2}
\end{array}\right)\,\,\,\,\,\,\,\mathrm{or}\,\,\,\,\,\,\left(\begin{array}{cccc}
1 & t & -1 & t\\
t & 1 & t & -1\\
-1 & t & 1 & t\\
t & -1 & t & 1
\end{array}\right)\label{eq:F4 and BC4}
\end{equation}
where $\Delta=\mathrm{diag}(1,x)$ is a diagonal, enphasing matrix
with $x={\rm i}\exp({\rm i}a)$, and where $t^{2}x=1$, $x,t\in\mathbb{T}$.
This orbit is $H_{2}$-reducible and it is the lowest order example
of an orbit of reducible Hadamard matrices. While the generic defect
of the orbit is 1, at the point $x=1$, i.e. for the matrix $F_{2}\otimes F_{2}$,
the defect takes the value 3 \cite{Krakow webguide}.

In order 5, all complex Hadamard matrices are equivalent to $F_{5}$,
an isolated Hadamard matrix.

In order 6, not only affine but also non-affine orbits have been found.
The three-parameter, non-affine orbit $K_{6}^{(3)}$ \cite{Karlsson K_6^(3)}
contains the affine suborbit $F_{6}^{(2)}$ but also all previously
described, smaller non-affine orbits ($B_{6}^{(1)}$ \cite{Beau_Nic},
$M_{6}^{(1)}$\cite{Mat Szo}, $K_{6}^{(2)}$ \cite{Karlsson_JMP},
$X_{6}^{(2)}$ and $X_{6}^{(2)T}$ \cite{Szollosi X_6^(2)}). It is
itself a suborbit of a larger but not yet fully understood non-affine
orbit $G_{6}^{(4)}$\cite{Szollosi G_6^(4)}. Furthermore, $K_{6}^{(3)}$
exhausts the set of $H_{2}$-reducible matrices of order 6, making
the general elements of $G_{6}^{(4)}$ together with the isolated
matrix $S_{6}^{(0)}$ stand out as the lowest, composite order Hadamard
matrices that are not reducible. In this order, a few circulant \cite{Bjorck Froberg,Haagerup 1997}
or $BCCB$ Hadamard matrices exist, e.g. $C_{2}\otimes C_{3}$, but
there are no orbits of such matrices \cite{Karlsson_unpublished}.
The generic defect of $K_{6}^{(3)}$ and $F_{6}^{(2)}$ is 4.

In order 9, which is the case of main interest here, the Fourier orbit
$F_{9}^{(4)}(a,b,c,d)$ \cite{Krakow webguide} is the only orbit
found until now. It is an affine orbit with generic defect 4, and
it can be written on the equivalent and manifestly reducible form
\begin{equation}
F_{9}^{(4)}(a,b,c,d)\sim\left(\begin{array}{ccc}
F_{3} & \Delta_{1}F_{3} & \Delta_{2}F_{3}\\
F_{3} & \omega\Delta_{1}F_{3} & \omega^{2}\Delta_{2}F_{3}\\
F_{3} & \omega^{2}\Delta_{1}F_{3} & \omega\Delta_{2}F_{3}
\end{array}\right)\label{eq:F_9^(4)}
\end{equation}
where $\Delta_{1}=\mathrm{diag}(1,x_{1},x_{2})$, $\Delta_{2}=\mathrm{diag}(1,x_{3},x_{4})$,
$x_{1}=\omega_{9}e^{{\rm i}a}$, $x_{2}=\omega_{9}^{2}e^{{\rm i}c}$,
$x_{3}=\omega_{9}^{2}e^{{\rm i}b}$ and $x_{4}=\omega_{9}^{4}e^{{\rm i}d}$
with $\omega_{9}=\exp(2\pi{\rm i}/9)$. For later reference, note
that the Kronecker product $F_{3}\otimes F_{3}$ is equivalent to
an element of $F_{9}^{(4)}$ ($\Delta_{1}=\Delta_{2}=I$ in (\ref{eq:F_9^(4)})),
but also to the symmetric block circulant with circulant blocks matrix
\begin{equation}
C_{3}\otimes C_{3}=\left(\begin{array}{ccc}
C_{3} & \omega C_{3} & \omega C_{3}\\
\omega C_{3} & C_{3} & \omega C_{3}\\
\omega C_{3} & \omega C_{3} & C_{3}
\end{array}\right)\sim F_{3}\otimes F_{3}\label{eq:C_3_3}
\end{equation}
and these matrices have defect 16.

The orbit $F_{9}^{(4)}$ has a two-parameter suborbit of circulant
Hadamard matrices $FB_{9}^{(2)}=\mathrm{circ}(1,u,v,1,\omega_{3}u,\omega_{3}^{2}v,1,\omega_{3}^{2}u,\omega_{3}v)$
\cite{Backelin,Faugere 2001} with $u,v\in\mathbf{\mathbb{T}}$ related
to the parameters of $F_{9}^{(4)}$ in (\ref{eq:F_9^(4)}) through
$x_{1}=u\bar{v}^{2}$, $x_{2}=\bar{u}\bar{v}$ with $x_{3}=\omega_{3}^{2}x_{2}$
and $x_{4}=\omega_{3}^{2}x_{2}/x_{1}$. Like $F_{9}^{(4)}$, this
Backelin orbit has generic defect 4. It has 9 circulant one-parameter
suborbits with generic defect 6, for $v=\omega_{9}^{1+3n}u^{2}$,
$u=\omega_{9}^{1+3n}v$, $uv=\omega_{9}^{1+3n}$, $n=0,1,2$. At the
27 points where these orbits intersect, the matrices have defect 10.

The orbit $F_{9}^{(4)}$ has also a one-parameter suborbit of symmetric
$BCCB$ Hadamard matrices (see $BC_{9B}^{(1)}$ below), equivalent
to the $F_{9}^{(4)}$ of (\ref{eq:F_9^(4)}) with $x_{4}=x_{1}$ and
$x_{2}=x_{3}=x_{1}^{2}$, and with generic defect equal to 10.

Among other suborbits, there is a one-parameter orbit of self-adjoint
Hadamard matrices (\cite{Szollosi Thesis}, Prop. 3.4.12), with generic
defect equal to 12.

Two isolated matrices of order 9 are also known: the matrix $N_{9}^{(0)}$\cite{Beau_Nic,Krakow webguide}
is equivalent to a symmetric, not reducible matrix, while the recently
described matrix $S_{9}^{(0)}$ \cite{McNulty Weigert} is reducible\footnote{As pointed out by a referee, a matrix $Q_{9}$ and a one-parameter
orbit $Z_{9}(x)$ are reported in \cite{Szollosi Master Thesis}.
$Q_{9}$ is an isolated matrix neither equivalent to $N_{9}^{(0)}$
nor to $S_{9}^{(0)}$, while $Z_{9}(x)$, which has the generic defect
2, is equivalent to a non-affine suborbit of the full $BC_{9}^{(2)}$
orbit of the present paper.}. Two other matrices, $B_{9}^{(0)}$ \cite{Beau_Nic,Krakow webguide}
and $W_{9A}$\cite{Szollosi Exotic}, will be identified below as
equivalent to elements in a new non-affine and not reducible $BCCB$
orbit $BC_{9}^{(2)}$ \footnote{The finding that $B_{9}^{(0)}$ is not isolated but an element in
a two-parameter orbit was anticipated in \cite{Beau_Nic} based on
the observation that its defect is 2.}.

\section{Numerical experiments\label{sec:Numerical-experiments}}

Numerical experiments have indicated that there exists an orbit of
Hadamard matrices of order 9 on the form (all elements of absolute
value 1, $x,y,u,w\in\mathbb{T}$) 
\begin{equation}
H=\left(\begin{array}{ccccccccc}
1 & x & x & y & u & w & y & w & u\\
x & 1 & x & w & y & u & u & y & w\\
x & x & 1 & u & w & y & w & u & y\\
y & w & u & 1 & x & x & y & u & w\\
u & y & w & x & 1 & x & w & y & u\\
w & u & y & x & x & 1 & u & w & y\\
y & u & w & y & w & u & 1 & x & x\\
w & y & u & u & y & w & x & 1 & x\\
u & w & y & w & u & y & x & x & 1
\end{array}\right)\label{eq:Basic form}
\end{equation}
A matrix of this form has several outstanding properties. 
\begin{enumerate}
\item Any permutation of $(xyuw)$ results in a matrix that, after permutations
of rows and columns, coincides with the original one. 
\item It is symmetric. 
\item It is block circulant, with circulant blocks ($BCCB)$, 
\begin{equation}
H=\left(\begin{array}{ccc}
A & B & B^{T}\\
B^{T} & A & B\\
B & B^{T} & A
\end{array}\right)\label{eq:Basic block form}
\end{equation}
with 
\begin{equation}
A=\left(\begin{array}{ccc}
1 & x & x\\
x & 1 & x\\
x & x & 1
\end{array}\right)\,\,\,\,\,\,\,\,\mathrm{and}\,\,\,\,\,\,\,\,\,\,B=\left(\begin{array}{ccc}
y & u & w\\
w & y & u\\
u & w & y
\end{array}\right).\label{eq:Basic blocks A and B}
\end{equation}

\end{enumerate}
\bigskip{}
 \bigskip{}
 Let $\sigma=x+y+u+w$. It follows from the $BCCB$ property that
\begin{equation}
\Lambda=\frac{1}{9}(F_{3}\otimes F_{3})H(F_{3}\otimes F_{3})^{\dagger}=\left(\begin{array}{c}
1+2\sigma\hspace{18ex}0\\
\hspace{4ex}1-\sigma+3y\hspace{16ex}\\
\hspace{6ex}1-\sigma+3y\hspace{14ex}\\
\hspace{8ex}1-\sigma+3x\hspace{12ex}\\
\hspace{10ex}1-\sigma+3w\hspace{10ex}\\
\hspace{12ex}1-\sigma+3u\hspace{8ex}\\
\hspace{14ex}1-\sigma+3x\hspace{6ex}\\
\hspace{16ex}1-\sigma+3u\hspace{4ex}\\
0\hspace{17ex}1-\sigma+3w
\end{array}\right)\label{eq:eigenvalues}
\end{equation}
is a diagonal matrix with entries equal to the eigenvalues of $H$,
and that any two matrices of the general form (\ref{eq:Basic form})
commute.

For $H$ in (\ref{eq:Basic form}) to be Hadamard, the elements $x,y,u,w\in\mathbb{T}$
cannot all be chosen independently. By enforcing the unitarity constraint,
the missing relations between $x$, $y$, $u$ and $w$ can be found,
and the most general orbit of complex Hadamard matrices compatible
with (\ref{eq:Basic form}) can be constructed. As will be shown below,
the result is either a non-affine, two-parameter $BCCB$ orbit, to
be denoted $BC_{9}^{(2)}$, or a set of matrices all equivalent to
the $BCCB$ matrix $C_{3}\otimes C_{3}$ of (\ref{eq:C_3_3}).

\section{The unitarity constraint}

The unitarity constraint $H^{\dagger}H=9I$ on the matrix $H$ in
(\ref{eq:Basic form}) is most easily imposed by requiring that the
eigenvalues in (\ref{eq:eigenvalues}) have modulus 3, 
\begin{eqnarray}
4\sigma\bar{\sigma}+2(\sigma+\bar{\sigma})+1 & = & 9\nonumber \\
3x(1-\bar{\sigma})+3\bar{x}(1-\sigma)+(1-\bar{\sigma})(1-\sigma) & = & 0\label{unitarity cond}
\end{eqnarray}
where in the last equation the condition $x\bar{x}=1$ has been used.
There are three additional equations where $x$ is replaced by $y$,
$u$ or $w$.

If $\sigma=x+y+u+w=1$, all the unitarity conditions (\ref{unitarity cond})
are satisfied, with no additional constraints on $x,y,u,w\in\mathbb{T}$,
and the eigenvalue matrix is simply $\Lambda=3\times\mathrm{diag}(1,y,y,x,w,u,x,u,w)$.

Otherwise, if $\sigma\ne1$, the second equation in (\ref{unitarity cond})
reads $\mathcal{R}e(3x/(1-\sigma))=-1/2$ or, again using that $x\bar{x}=1$,
\begin{equation}
\frac{3x}{1-\sigma}=-\frac{1}{2}\pm{\rm i}\sqrt{\frac{9}{|1-\sigma|^{2}}-\frac{1}{4}}\label{eq:unit cond x}
\end{equation}
with similar expressions for $y$, $u$ and $w$. Adding these relations,
and taking into account all sign combinations for the square root
terms, one finds five possible conditions on $\sigma$, 
\[
\frac{3\sigma}{1-\sigma}=-2+\left\{ \begin{array}{c}
4\\
2\\
0\\
-2\\
-4
\end{array}\right\} {\rm i}\sqrt{\frac{9}{|1-\sigma|^{2}}-\frac{1}{4}}
\]
In the first and last cases, 
\[
\frac{3|\sigma|}{|1-\sigma|}=4\frac{3}{|1-\sigma|}
\]
i.e. $|\sigma|=4$. Such sigmas are not compatible with the first
of the equations (\ref{unitarity cond}). Similarly, in the second
and fourth cases, $|2\sigma+1|=3\sqrt{3}$, again not allowed by (\ref{unitarity cond}).
In the third case, finally, $\sigma$ equals $-2$. This is a value
compatible with (\ref{unitarity cond}) and implies (see (\ref{eq:unit cond x}))
that $x$ equals $\omega$ or $\omega^{2}$, and correspondingly for
$y,\,u$ and $w$. Taking into account that the sum of $x,\,y,\,u$
and $w$ should equal $-2$, the final result is that two of these
parameters must have the value $\omega$, and the other two $\omega^{2}$.
For the corresponding eigenvalue matrix one finds $\Lambda=\mathrm{3\times diag}(-1,1+y,1+y,1+x,1+w,1+u,1+x,1+u,1+w)$.
Recall that $x=\omega$ implies $1+x=-\omega^{2}$ etc, i.e. in addition
to $-3$ there are four eigenvalues equal to $-3\omega$ and four
equal to $-3\omega^{2}$.

The above findings are collected in 
\begin{thm}
\label{thm:Main theorem}For the elements $x,y,u,w\in\mathbb{T}$
of $H$ in (\ref{eq:Basic form}), the unitarity condition $H^{\dagger}H=9I$
implies that either 
\[
x+y+u+w=1
\]
or that one pair of $x,y,u,w$ equals $\omega$ and the other pair
equals $\omega^{2}$, where $\omega=\exp(2\pi{\rm i}/3)$. In this
last case 
\[
x+y+u+w=-2.
\]

\end{thm}
\bigskip{}

In the $x+y+u+w=1$ case, the interdependence between $x,y,u$ and
$w$ can be illustrated with the help of a complex parameter $\zeta$: 
\begin{lem}
Let $x+y+u+w=1$ with $x,y,u,w\in\mathbb{T}$, and let $\zeta=2(x+y)-1=-2(u+w)+1$.
Then, for $\zeta\ne\pm1$,

1. $\zeta$ is a complex number in the intersection of the circular
discs $|1+\zeta|\le4$ and $|1-\zeta|\le4$

2. $x=\frac{1}{4}(1+\zeta)(1+{\rm i}\sqrt{\frac{16}{|1+\zeta|^{2}}-1})$
and $y=\frac{1}{4}(1+\zeta)(1-{\rm i}\sqrt{\frac{16}{|1+\zeta|^{2}}-1})$
or vice versa

3. $u=\frac{1}{4}(1-\zeta)(1+{\rm i}\sqrt{\frac{16}{|1-\zeta|^{2}}-1})$
and $w=\frac{1}{4}(1-\zeta)(1-{\rm i}\sqrt{\frac{16}{|1-\zeta|^{2}}-1})$
or vice versa\end{lem}
\begin{proof}
If $z_{1}$ and $z_{2}$ are of modulus 1, and $z_{1}+z_{2}=c$, then,
for any $0<|c|\le2$, $z_{1}=\frac{c}{2}(1\pm{\rm i}\sqrt{\frac{4}{|c|^{2}}-1})$
and $z_{2}=\frac{c}{2}(1\mp{\rm i}\sqrt{\frac{4}{|c|^{2}}-1})$. 
\end{proof}
For each $\zeta$ in the allowed region there are therefore 4 different,
but equivalent Hadamard matrices. In the exceptional point where $\zeta=1$,
either $x=-\omega$ and $y=-\omega^{2}$, or vice versa, with $u=-w$
arbitrary in $\mathbb{T}$. This single point in the $\zeta$ parameter
space therefore corresponds to two instances of a one-parameter set.
This phenomenon is an artifact of the $\zeta$ parametrization that
also manifests itself in the above expressions for $u$ and $v$:
with $\zeta=1+\epsilon e^{{\rm i}\phi}$, the limit $\epsilon\to0^{+}$
depends on the parameter $\phi$. The situation where $\zeta\to-1$
is analogous.

\section{The new non-affine orbit $BC_{9}^{(2)}$and the set $BC_{9}(-2)$}

In terms of orbits of Hadamard matrices, the results of the previous
section can be summarized as follows. The full set of symmetric $BCCB$
Hadamard matrices $BC_{9}(x,y,u,w)$ on the form (\ref{eq:Basic form})
can be seen as 4 instances ($x\leftrightarrow y$ and/or $u\leftrightarrow w$)
of a two parameter orbit $BC_{9}^{(2)}(\zeta)$ where 
\begin{eqnarray}
\left.\begin{array}{c}
x\\
y
\end{array}\right\}  & = & \frac{1}{4}(1+\zeta)(1\pm{\rm i}\sqrt{\frac{16}{|1+\zeta|^{2}}-1})\nonumber \\
\left.\begin{array}{c}
u\\
w
\end{array}\right\}  & = & \frac{1}{4}(1-\zeta)(1\pm{\rm i}\sqrt{\frac{16}{|1-\zeta|^{2}}-1})\label{eq:xyuw_zeta-1-1}
\end{eqnarray}
In addition there are six different but equivalent matrices where
one pair of $x$, $y$, $u$ and $w$ is equal to $\omega$, and the
other pair equals $\omega^{2}$, so that $x+y+u+w=-2$. These matrices
will collectively be denoted $BC_{9}(-2)$.

No other complex Hadamard matrices are compatible with the form (\ref{eq:Basic form}).

The orbit $BC_{9}^{(2)}$ is a non-affine orbit of symmetric $BCCB$
complex Hadamard matrices, and the defect of a generic element is
2. Since the value of the defect coincides with the number of continuous
parameters, $BC_{9}^{(2)}$ is not contained in any orbit with additional
continuous parameters, and it is in particular not a suborbit of the
Fourier orbit $F_{9}^{(4)}$ (for which the defect of a generic element
is 4).

The six $BC_{9}(-2)$ matrices are all equivalent to the matrix $C_{3}\otimes C_{3}$
(or $F_{3}\otimes F_{3},$ a matrix which is also in the Fourier orbit
$F_{9}^{(4)}$), and their defect is 16.

\section{Non-reducibility of $BC_{9}^{(2)}$}

For the new orbit $BC_{9}^{(2)}$ constructed in the previous sections,
the question of reducibility has not yet been addressed. A Hadamard
matrix of order 9 can at most be $H_{3}$-reducible, and if this is
the case, it must be equivalent to a dephased Hadamard matrix where
the upper left $3\times3$ submatrix is $F_{3}$. If the new orbit
were $H_{3}$-reducible, in the dephased form there would be at least
two elements $\omega$, and two elements $\omega^{2}$, in such positions
that an $F_{3}$ submatrix can be generated through equivalence transformations.
For arbitrary parameters $x,y\in\mathbb{T}$, this cannot happen.
Indeed, there are $9\times9=81$ instances of dephased matrices, and
they all have elements taken from the set $\{1,a,\frac{1}{a},a^{2},\frac{1}{a^{2}},\frac{a}{b},\frac{a^{2}}{b^{2}},\frac{ab}{c},\frac{a}{bc},\frac{ab}{c^{2}},\frac{c^{2}}{ab}\}$
where $a,b$ and $c$ are (all different and) equal to any combination
of $x,y,u$ and $w$. No member in the set equals to $\omega$ or
$\omega^{2}$ over the entire parameter space, implying that the orbit
$BC_{9}^{(2)}$ is not reducible. For special values of the parameters,
on the other hand, $H_{3}$-reducibility may occur, as exemplified
by the appearance of the matrix $C_{3}^{\dagger}\otimes C_{3}$ for
$x=\omega,\,y=\omega^{2},\,u=w=1$.

\section{Affine suborbits of $BC_{9}^{(2)}$ \label{sec:Affine-suborbits-of}}

One-parameter suborbits of $BC_{9}^{(2)}$ will in general be non-affine.
However, two cases of affine suborbits can be identified.

\smallskip{}

\subsection{The affine suborbit $BC_{9A}^{(1)}$ }

~\smallskip{}
 ~\\
If $(xyuw)$ is a permutation of $(\mu,-\mu,-\omega,-\omega^{2}),$
with $\mu\in\mathbb{T}$, then $\zeta=\pm1$, ${\rm i}\sqrt{3}\pm2\mu$
or $-{\rm i}\sqrt{3}\pm2\mu$. The last two possibilities correspond
to two cases of two overlapping circles in the complex $\zeta$ plane,
of radius 2 and with midpoints at $\pm{\rm i}\sqrt{3}$ (Figure \ref{fig:Affine-suborbits-of}).
The situation at the points $\zeta=\pm1$ is more subtle since the
$\mu$ degree of freedom is hidden in how these points are approached
in the complex plane. Indeed, let for instance $\zeta=-1-\epsilon{\rm i}\mu$
, $\mu\in\mathbb{T}$, and take $\epsilon\to0^{+}$ in (\ref{eq:xyuw_zeta-1-1}).
As a result, $x,y\to\pm\mu$ and $u,w\to-\omega^{2}\mathrm{\,\,or\,\,}-\omega$,
as advertized. As already mentioned, this behavior is an artifact
of the parametrization (\ref{eq:xyuw_zeta-1-1}), and has been noted
before in a similar case \cite{Karlsson_JMP}.

Taking $(xyuw)=(\mu,-\mu,-\omega,-\omega^{2})$, the resulting suborbit
$BC_{9A}^{(1)}$ has the representation (\ref{eq:Basic block form})
with 
\begin{equation}
A=\left(\begin{array}{ccc}
1 & \mu & \mu\\
\mu & 1 & \mu\\
\mu & \mu & 1
\end{array}\right)\,\,\,\,\,\,\,\,\mathrm{and}\,\,\,\,\,\,\,\,\,\,B=-\left(\begin{array}{ccc}
\mu & \omega & \omega^{2}\\
\omega^{2} & \mu & \omega\\
\omega & \omega^{2} & \mu
\end{array}\right).\label{eq:C_9A^(1), A and B}
\end{equation}
The defect of a generic element of $BC{}_{9A}^{(1)}$ is 2. Interestingly,
although affine, $BC_{9A}^{(1)}$ is not a suborbit of the affine
orbit $F_{9}^{(4)}$. 
\begin{figure}
\includegraphics[bb=100bp 250bp 500bp 550bp,clip,width=0.4\paperwidth]{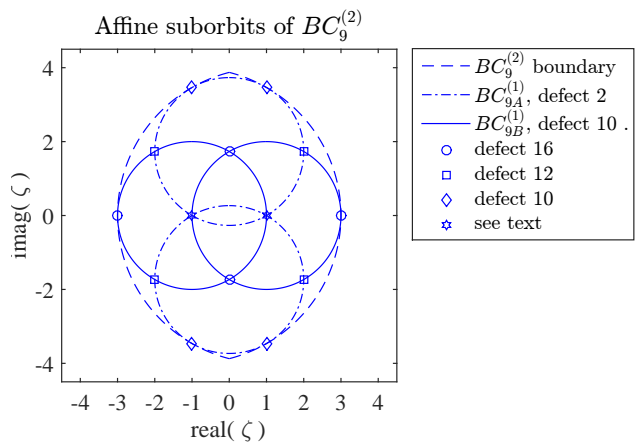}

\caption{Affine suborbits of $BC_{9}^{(2)}$, and the defect at intersection
points \label{fig:Affine-suborbits-of}}
\end{figure}

\smallskip{}

\subsection{The affine suborbit $BC_{9B}^{(1)}$}

~\smallskip{}
 ~\\
If $(xyuw)$ is a permutation of $(1,\xi,\omega\xi,\omega^{2}\xi),$
with $\xi\in\mathbb{T}$, then $\zeta=\pm(1+2\xi)$, $\pm(1+2\omega\xi)$
or $\pm(1+2\omega^{2}\xi)$ and there are two cases of three overlapping
circles in the complex $\zeta$ plane, again of radius 2 but with
midpoints at $\pm1$ (Figure \ref{fig:Affine-suborbits-of}). Taking
$(xyuw)=(1,\xi,\omega\xi,\omega^{2}\xi)$, the resulting affine suborbit
$BC_{9B}^{(1)}$ has the representation (\ref{eq:Basic block form})
with

\begin{equation}
A=\left(\begin{array}{ccc}
1 & 1 & 1\\
1 & 1 & 1\\
1 & 1 & 1
\end{array}\right)\,\,\,\,\,\,\,\,\mathrm{and}\,\,\,\,\,\,\,\,\,\,B=\xi\left(\begin{array}{ccc}
1 & \omega & \omega^{2}\\
\omega^{2} & 1 & \omega\\
\omega & \omega^{2} & 1
\end{array}\right),\label{eq:C_9B^(1), A and B-1}
\end{equation}
It is straightforward to show that the matrices of $BC_{9B}^{(1)}$
are equivalent to Fourier matrices on the form (\ref{eq:F_9^(4)})
with $\Delta_{1}=\mathrm{diag}(1,\bar{\xi},\bar{\xi}^{2})$ and $\Delta_{2}=\mathrm{diag}(1,\bar{\xi}^{2},\bar{\xi})$,
i.e. that the $BC_{9B}^{(1)}$ orbits can be seen as one-parameter
intersections of $BC_{9}^{(2)}$ with the Fourier orbit $F_{9}^{(4)}$.
The defect of a generic element of $BC{}_{9B}^{(1)}$ turns out to
be 10.

\medskip{}

\begin{conjecture}
$BC_{9A}^{(1)}$ and $BC{}_{9B}^{(1)}$ are the only affine suborbits
of $BC_{9}^{(2)}$. 
\end{conjecture}
\medskip{}

As support for this conjecture, consider the graphs in the complex
plane representing the relation $x+y+u+w=1$ characteristic for $BC_{9}^{(2)}$.
In the $BC_{9A}^{(1)}$ case two links add up to zero and can for
that reason have any orientation (the $\mu$ parameter). In the $BC_{9B}^{(1)}$
case three links add up to zero and the resulting triangle can have
any orientation (the $\xi$ parameter). We see no other way for a
subgraph to exhibit a similar rotational invariance, and hence no
room for an associated affine parameter.

\section{One-parameter suborbits with defect 4 }

The occurrence of a full suborbit with the anomalous defect 10, the
\textit{$BC{}_{9B}^{(1)}$} of the previous section, has motivated
the search for additional suborbits with a defect different from the
generic value 2. A numerical search has revealed 24 one-variable suborbit
pieces in $BC_{9}^{(2)}$ with the anomalous defect 4 (Figure \ref{fig:Defect-4-suborbits}).
Lacking an analytic description, we have chosen not to try to put
these pieces together into full suborbits. However, if the two generic
pieces\footnote{based on $\sim700$ points.} shown in Figure \ref{fig:Defect-4-generic}
are complemented by other instances of themselves, obtained by permuting
the $(xyuw)$ parameters, one half of the full suborbit structure
in Figure \ref{fig:Defect-4-suborbits} is obtained. The remaining
half follows by taking the mirror image in the real $\zeta$-axis.
\begin{figure}
\includegraphics[bb=100bp 250bp 500bp 550bp,clip,width=0.4\paperwidth]{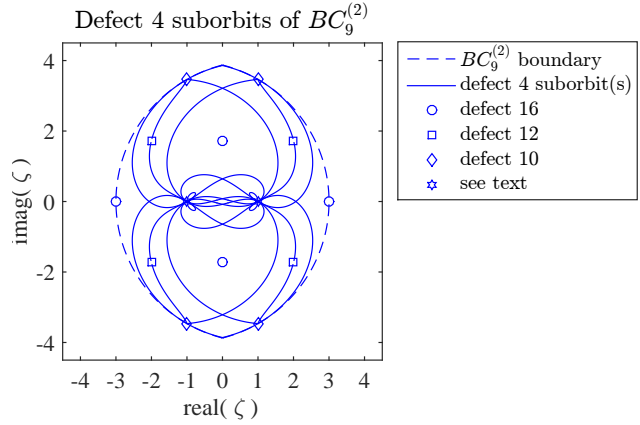}\caption{Defect 4 suborbits of $BC_{9}^{(2)}$\label{fig:Defect-4-suborbits}}
\end{figure}

\begin{figure}
\includegraphics[bb=50bp 250bp 600bp 530bp,clip,width=0.5\paperwidth]{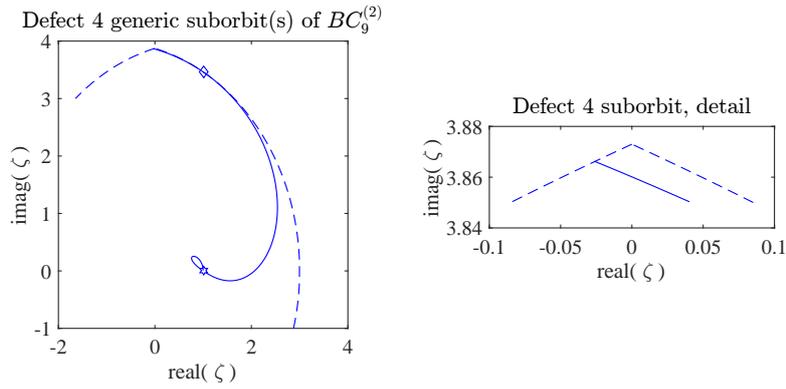}\caption{Defect 4 generic suborbit(s) of $BC_{9}^{(2)}$\label{fig:Defect-4-generic}}
\end{figure}

\section{Special matrices}

Of some interest are the points where the $BC_{9A}^{(1)}$ and $BC_{9B}^{(1)}$
orbits intersect each other, other instances of themselves, or the
boundary of the parameter domain, see Figure \ref{fig:Affine-suborbits-of}.

The $BC_{9B}^{(1)}$ orbit intersects another such orbit at the points
$\zeta=\pm{\rm i}\sqrt{3}$, and the boundary at $\zeta=\pm3$. The
corresponding matrices are equivalent to $C_{3}\otimes C_{3}\sim F_{3}\otimes F_{3}$,
see equation (\ref{eq:C_3_3}). This matrix is in the Butson class
$BH(9,3)$, and its defect is known to be 16 \cite{Krakow webguide}.

A\textbf{ }$BC_{9B}^{(1)}$ orbit intersects a $BC_{9A}^{(1)}$ orbit
at the six points $\zeta=\pm1$, $\pm(1-2\omega)$ and $\pm(1-2\omega^{2})$.
For instance, with $\mu=1$ or $\xi=-1$, $(xyuw)$ equals $(1,-1,-\omega,-\omega^{2})$
and $\zeta=-1$. The resulting matrix $BC_{9A\cap B}$ is a symmetric,
block circulant with circulant blocks matrix in the Butson class $BH(9,6)$,

\begin{equation}
BC_{9A\cap B}=\left(\begin{array}{ccc}
A & B & B^{T}\\
B^{T} & A & B\\
B & B^{T} & A
\end{array}\right)\label{eq:C_9AB^(0)}
\end{equation}
with 
\begin{equation}
A=\left(\begin{array}{ccc}
1 & 1 & 1\\
1 & 1 & 1\\
1 & 1 & 1
\end{array}\right)\,\,\,\,\,\,\,\,\mathrm{and}\,\,\,\,\,\,\,\,\,\,B=-\left(\begin{array}{ccc}
1 & \omega & \omega^{2}\\
\omega^{2} & 1 & \omega\\
\omega & \omega^{2} & 1
\end{array}\right),\label{eq:C_9AB^(0), A and B}
\end{equation}
and it is equivalent to the (Fourier orbit) matrix (\ref{eq:F_9^(4)})
with $\Delta_{1}=\mathrm{diag}(1,-1,1)$ and $\Delta_{2}=\mathrm{diag}(1,1,-1)$.
In whichever form, this matrix is a natural, common seed matrix for
all three affine orbits $F_{9}^{(4)}$, $BC{}_{9A}^{(1)}$ and $BC_{9B}^{(1)}$,
and its defect is 12.

Instances of the $BC_{9A}^{(1)}$ orbit intersect the domain boundary
at the four points $\zeta=\pm(1+4\omega)$ and $\zeta=\pm(1+4\omega^{2})$.
Taking $(xyuw)=(\omega,-\omega^{2,}-\omega,-\omega)$ at $\zeta=1+4\omega$
results in the $BCCB$ matrix, 
\begin{equation}
BC_{9Ab}=\left(\begin{array}{ccc}
C_{3} & -\omega^{2}C_{3}^{\dagger} & -\omega^{2}C_{3}^{\dagger}\\
-\omega^{2}C_{3}^{\dagger} & C_{3} & -\omega^{2}C_{3}^{\dagger}\\
-\omega^{2}C_{3}^{\dagger} & -\omega^{2}C_{3}^{\dagger} & C_{3}
\end{array}\right)
\end{equation}
Similarly, taking $(xyuw)=(\omega^{2},-\omega,-\omega^{2},-\omega^{2})$
at $\zeta=1+4\omega^{2}$ results in $BC_{9Ab}^{\dagger}$. These
two matrices are both in the Butson class $BH(9,6)$ and have defect
10. A direct calculation has shown that they are not equivalent. 
\begin{prop}
The Butson matrices $BC_{9Ab}$ and $BC_{9Ab}^{\dagger}$ are not
elements in the Fourier orbit $F_{9}^{(4)}$.\end{prop}
\begin{proof}
The two matrices $BC_{9Ab}$ and $BC_{9Ab}^{\dagger}$ are in $BH(9,6)$
and have defect 10. Taking $x_{1}$, $x_{2}$, $x_{3}$ and $x_{4}$
in (\ref{eq:F_9^(4)}) as any combination of $\pm1$, $\pm\omega$
and $\pm\omega^{2}$, a direct evaluation shows that for the resulting
$6^{4}$ $BH(9,6)$ matrices in $F_{9}^{(4)}$ the defect is $4$
in 864 cases, $8$ in 243 cases, $12$ in 162 cases and $16$ in 27
cases, but never $10$. 
\end{proof}
Finally, two complex Hadamard matrices of order 9 reported earlier
can now be identified as elements in $BC_{9}^{(2)}$. For $(xyuw)$
some permutation of $(\tau\tau\bar{\tau}\bar{\tau})$, with $\tau=\frac{1}{4}(1+{\rm i}\sqrt{15})$,
either $\zeta=0$ or $\zeta=\pm{\rm i}\sqrt{15}$. The corresponding
element of $BC_{9}^{(2)}$ is equivalent to the matrix $W_{9A}$ of
\cite{Szollosi Exotic}. For $(xyuw)$ some permutation of $(\epsilon\,\epsilon^{3}\epsilon^{7}\epsilon^{9})$,
with $\epsilon=\exp(2\pi{\rm i}/10)$, the corresponding elements
of $BC_{9}^{(2)}$ are equivalent to the $B_{9}^{(0)}$ of \cite{Beau_Nic,Krakow webguide}.
This matrix is found at the six points where $\zeta=\pm\sqrt{5}$,
$\zeta=\pm{\rm i}(\sqrt{(5+\sqrt{5})/2}+\sqrt{(5-\sqrt{5})/2})$ or
$\zeta=\pm{\rm i}(\sqrt{(5+\sqrt{5})/2}-\sqrt{(5-\sqrt{5})/2})$.
The defect for the matrices $B_{9}^{(0)}$ and $W_{9A}$ is 2.

\section{MUBs for $N=3^{2}$ \label{sec:MUBs}}

Mutually unbiased bases (MUBs) in $\mathbb{C}^{N}$ are closely related
to complex Hadamard matrices: taking the basis vectors as columns
of a matrix, the standard basis can be represented by the unit matrix,
and all other bases will then appear as (mutually unbiased and) in
general enphased complex Hadamard matrices. This section serves to
illustrate how the results of the present paper are related to such
bases in dimension $N=9$.

Complete sets of $N+1$ MUBs exist for all prime and powers of a prime
dimensions $N$ (for recent reviews, see \cite{Godsil Roy 2009,Durt et al}).
For example, for $N=9=3^{2}$, the complete set $\{B_{i}\}_{i=0...9}$
of MUBs given in \cite{Wootters_Fields,Chaturvedi} have the form
$B_{0}=I$ and $B_{i}=\frac{1}{3}D_{i}(F_{3}\otimes F_{3})$ where
the unitary, diagonal matrices $D_{i}$ are specified in Table \ref{tab:Diagonals-of-D_i}.
The $D_{i}$ matrices form a closed set under the product $D_{i}^{\dagger}D_{j}=D_{k}$
according to the pattern displayed for the $M_{i}$ matrices in Table
\ref{tab:M_i_dag_M_j}. 
\begin{table}
\begin{tabular}{|c|c|c|c|c|c|c|c|c|}
\hline 
$\,\,D_{1}\,\,$  & $\,\,D_{2}\,\,$  & $\,\,D_{3}\,\,$  & $\,\,D_{4}\,\,$  & $\,\,D_{5}\,\,$  & $\,\,D_{6}\,\,$  & $\,\,D_{7}\,\,$  & $\,\,D_{8}\,\,$  & $\,\,D_{9}\,\,$\tabularnewline
\hline 
\hline 
$1$  & $1$  & $1$  & $1$  & $1$  & $1$  & $1$  & $1$  & $1$\tabularnewline
\hline 
$1$  & $\omega^{2}$  & $\omega$  & $\omega$  & $1$  & $\omega^{2}$  & $\omega^{2}$  & $\omega$  & $1$\tabularnewline
\hline 
$1$  & $\omega^{2}$  & $\omega$  & $\omega$  & $1$  & $\omega^{2}$  & $\omega^{2}$  & $\omega$  & $1$\tabularnewline
\hline 
$1$  & $1$  & $1$  & $\omega$  & $\omega$  & $\omega$  & $\omega^{2}$  & $\omega^{2}$  & $\omega^{2}$\tabularnewline
\hline 
$1$  & $\omega$  & $\omega^{2}$  & $\omega^{2}$  & $1$  & $\omega$  & $\omega$  & $\omega^{2}$  & $1$\tabularnewline
\hline 
$1$  & $1$  & $1$  & $\omega^{2}$  & $\omega^{2}$  & $\omega^{2}$  & $\omega$  & $\omega$  & $\omega$\tabularnewline
\hline 
$1$  & $1$  & $1$  & $\omega$  & $\omega$  & $\omega$  & $\omega^{2}$  & $\omega^{2}$  & $\omega^{2}$\tabularnewline
\hline 
$1$  & $1$  & $1$  & $\omega^{2}$  & $\omega^{2}$  & $\omega^{2}$  & $\omega$  & $\omega$  & $\omega$\tabularnewline
\hline 
$1$  & $\omega$  & $\omega^{2}$  & $\omega^{2}$  & $1$  & $\omega$  & $\omega$  & $\omega^{2}$  & $1$\tabularnewline
\hline 
\end{tabular}

\medskip{}
 \caption{Diagonals of the MUB unitaries $D_{i}$ as obtained from \cite{Chaturvedi}\label{tab:Diagonals-of-D_i}.
Note that $D_{3}=D_{2}^{\dagger}$, $D_{7}=D_{4}^{\dagger}$, $D_{8}=D_{6}^{\dagger}$
and $D_{9}=D_{5}^{\dagger}$.}
\end{table}

Several other complete sets of MUBs have been shown to be equivalent
to this set in the sense that one set can be obtained from another
by means of an overall unitary transformation, c.f. Ref \cite{Godsil Roy 2009}.

Of particular interest in the context of the present paper is the
complete set $\{M_{i}\}_{i=0...9}$ which is obtained from the $\{B_{i}\}$
set through left multiplication by (the unitary matrix) $\frac{1}{3}(F_{3}\otimes F_{3})^{\dagger}$,
i.e. \footnote{note that $(F_{3}\otimes F_{3})^{\dagger}$ and $F_{3}\otimes F_{3}$
are column permutation equivalent.} 
\begin{eqnarray*}
M_{0} & = & \frac{1}{3}(F_{3}\otimes F_{3})^{\dagger}\,\,\mathrm{or,\,\,equivalently,}\,\,\frac{1}{3}F_{3}\otimes F_{3}\\
M_{1} & = & I\\
M_{i} & = & \frac{1}{9}(F_{3}\otimes F_{3})^{\dagger}D_{i}(F_{3}\otimes F_{3}),\,\,\,i=2,...,9
\end{eqnarray*}
Except for $M_{0}$ and $M_{1}$, the $M_{i}$ matrices are symmetric,
\textit{block circulant with circulant blocks} mutually unbiased Hadamard
matrices, i.e. (apart from normalization) matrices of the general
form (\ref{eq:Basic form}) investigated in the previous sections.
The outstanding features \cite{Combescure II} of this set are that,
for $i,\,j,\,k$ in $\{1,2,...,9\}$,

\begin{eqnarray*}
M_{i}^{3} & = & I\\
M_{i}M_{j} & = & M_{j}M_{i}\\
M_{i}^{\dagger}M_{j} & = & M_{k},\,\,\mathrm{see\,\,Table\,\,\ref{tab:M_i_dag_M_j}},
\end{eqnarray*}
all properties inherited from the $\{D_{i}\}_{i=1,...,9}$ matrices.
Furthermore, $M_{3}=M_{2}^{\dagger}$, $M_{7}=M_{4}^{\dagger}$, $M_{8}=M_{6}^{\dagger}$
and $M_{9}=M_{5}^{\dagger}$, i.e. for each $M_{i}$ there is an $M{}_{j}$
such that $M_{i}^{\dagger}=M_{j}$. This set of MUBs is therefore
invariant under Hermitan conjugation. 
\begin{table}
\begin{tabular}{|c||c|c|c||c|c|c||c|c|c|}
\hline 
 & $\,M_{1}\,$  & $\,M_{2}\,$  & $\,M_{3}\,$  & $\,M_{4}\,$  & $\,M_{5}\,$  & $\,M_{6}\,$  & $\,M_{7}\,$  & $\,M_{8}\,$  & $\,M_{9}\,$\tabularnewline
\hline 
\hline 
$M_{1}^{\dagger}$  & $I$  & $M_{2}$  & $M_{3}$  & $M_{4}$  & $M_{5}$  & $M_{6}$  & $M_{7}$  & $M_{8}$  & $M_{9}$\tabularnewline
\hline 
$M_{2}^{\dagger}$  & $M_{3}$  & $I$  & $M_{2}$  & $M_{6}$  & $M_{4}$  & $M_{5}$  & $M_{9}$  & $M_{7}$  & $M_{8}$\tabularnewline
\hline 
$M_{3}^{\dagger}$  & $M_{2}$  & $M_{3}$  & $I$  & $M_{5}$  & $M_{6}$  & $M_{4}$  & $M_{8}$  & $M_{9}$  & $M_{7}$\tabularnewline
\hline 
\hline 
$M_{4}^{\dagger}$  & $M_{7}$  & $M_{8}$  & $M_{9}$  & $I$  & $M_{2}$  & $M_{3}$  & $M_{4}$  & $M_{5}$  & $M_{6}$\tabularnewline
\hline 
$M_{5}^{\dagger}$  & $M_{9}$  & $M_{7}$  & $M_{8}$  & $M_{3}$  & $I$  & $M_{2}$  & $M_{6}$  & $M_{4}$  & $M_{5}$\tabularnewline
\hline 
$M_{6}^{\dagger}$  & $M_{8}$  & $M_{9}$  & $M_{7}$  & $M_{2}$  & $M_{3}$  & $I$  & $M_{5}$  & $M_{6}$  & $M_{4}$\tabularnewline
\hline 
\hline 
$M_{7}^{\dagger}$  & $M_{4}$  & $M_{5}$  & $M_{6}$  & $M_{7}$  & $M_{8}$  & $M_{9}$  & $I$  & $M_{2}$  & $M_{3}$\tabularnewline
\hline 
$M_{8}^{\dagger}$  & $M_{6}$  & $M_{4}$  & $M_{5}$  & $M_{9}$  & $M_{7}$  & $M_{8}$  & $M_{3}$  & $I$  & $M_{2}$\tabularnewline
\hline 
$M_{9}^{\dagger}$  & $M_{5}$  & $M_{6}$  & $M_{4}$  & $M_{8}$  & $M_{9}$  & $M_{7}$  & $M_{2}$  & $M_{3}$  & $I$\tabularnewline
\hline 
\end{tabular}

\medskip{}
 \caption{Multiplication table for $M_{i}^{\dagger}M_{j}$\label{tab:M_i_dag_M_j}.
Note the block circulant with circulant blocks structure.}
\end{table}

The parameters in (\ref{eq:Basic form}) corresponding to the $\{M_{i}\}_{i=2,...,9}$
matrices are easily identified through inspection of the eigenvalues.
A comparison of Table \ref{tab:Diagonals-of-D_i} with the eigenvalues
(\ref{eq:eigenvalues}) for $\sigma=1$ and $\sigma=-2$ shows that
for $M_{2},M_{3},M_{5}$ and $M_{9}$ the parameters $(xyuw)$ in
(\ref{eq:Basic form}) are permutations of $(1,1,\omega,\omega^{2})$,
see Table \ref{tab:Parameters-in-the}. These matrices are therefore
equal to elements in the suborbit $BC_{9B}^{(1)}$ in Section \ref{sec:Affine-suborbits-of}
(the intersection of the $BC_{9}^{(2)}$ of Theorem \ref{thm:Main theorem}
with $F_{9}^{(4)}$) corresponding to $\xi=1$. For $-M_{4},-M_{6},-M_{7}$
and $-M_{8}$ the parameters $(xyuw)$ are permutations of $(\omega,\omega,\omega^{2},\omega^{2})$,
see Table \ref{tab:Parameters-in-the}, characteristic for matrices
in the exceptional set $BC_{9}(-2)$ also identified in Theorem \ref{thm:Main theorem}.

It should be pointed out that as Hadamard matrices, the $M_{i}$ matrices
for $i\ge2$ are all equivalent to each other, and to the $\frac{1}{3}C_{3}\otimes C_{3}$
matrix of (\ref{eq:C_3_3}) (for instance, $-M_{7}=\frac{1}{3}C_{3}\otimes C_{3}$).
Their common defect is therefore 16. 
\begin{table}
\begin{tabular}{|c|c|c|c|c|c|c|c|c|}
\hline 
 & $M_{2}$  & $M_{3}$  & $-M_{4}$  & $M_{5}$  & $-M_{6}$  & $-M_{7}$  & $-M_{8}$  & $M_{9}$\tabularnewline
\hline 
\hline 
$x$  & $1$  & $1$  & $\omega^{2}$  & $\omega$  & $\omega^{2}$  & $\omega$  & $\omega$  & $\omega^{2}$\tabularnewline
\hline 
$y$  & $\omega^{2}$  & $\omega$  & $\omega^{2}$  & $1$  & $\omega$  & $\omega$  & $\omega^{2}$  & $1$\tabularnewline
\hline 
$u$  & $1$  & $1$  & $\omega$  & $\omega^{2}$  & $\omega$  & $\omega^{2}$  & $\omega^{2}$  & $\omega$\tabularnewline
\hline 
$w$  & $\omega$  & $\omega^{2}$  & $\omega$  & $1$  & $\omega^{2}$  & $\omega^{2}$  & $\omega$  & $1$\tabularnewline
\hline 
\end{tabular}\medskip{}

\caption{Parameters in the $BC_{9}(x,y,u,w)$ orbit (see (\ref{eq:Basic form}))
corresponding to the $M_{i}$ matrices. Note that $u=\bar{x}$ and
$w=\bar{y}$. \label{tab:Parameters-in-the}}
\end{table}

As a final remark, recall from Section \ref{sec:Numerical-experiments}
that a permutation of the $(xyuw)$ parameters in (\ref{eq:Basic form})
can be undone by permutations of rows and columns. Specifically, a
permutation of the parameters $(xyuw)$ in Table \ref{tab:Parameters-in-the}
either leaves the set $\{M_{i}\}_{i=1,...,9}$ invariant, or gives
rise to one of two additional, but equivalent, sets $\{P_{1}'M_{i}P_{1}\}_{i=1,...,9}$
or $\{P_{2}'M_{i}P_{2}\}_{i=1,...,9}$ where $P_{1}$ and $P_{2}$
are (unitary) permutation matrices that interchange $y$ and $u$,
or $u$ and $w$, respectively.

Explicit expressions for complete sets of MUBs in 9 dimensions have
been published several times. It has been verified that the versions
given in references \cite{Aravind,Klimov et al MUB9,Wiesniak et al}\footnote{There is a misprint in Table 5 of \cite{Aravind}. The corrected entries
read $|3_{4}\rangle=1\,\bar{\omega}\,1\,\omega\,1\,\omega\,\omega\,1\,\omega$
and $|3_{5}\rangle=1\,\omega\,\omega\,\omega\,\bar{\omega\,}\bar{\omega\,}\omega\,\bar{\omega\,}\bar{\omega}$.
There is also a misprint in the $B_{3}$ of Appendix G in \cite{Wiesniak et al}.
The corrected matrix elements read $B_{3}(2,5)=\alpha_{3}^{2}$, $B_{3}(3,5)=$1
and $B_{3}(4,5)=\alpha_{3}$.} indeed are equivalent to the $\{M_{i}\}$ set of the present section
(and therefore also to the set of references \cite{Wootters_Fields,Chaturvedi}).

\subsection*{Acknowledgement}

Questions and suggestions by a referee are gratefully acknowledged.
We are particularly indebted to her/him for providing us with a copy
of the not generally available reference \cite{Szollosi Master Thesis}.

\newpage{} 

\end{document}